\documentclass{article}

\usepackage{times}
\usepackage{graphicx} 
\usepackage{color}
\usepackage[fleqn]{amsmath}
\usepackage{amssymb, amsmath, amsthm}

\usepackage[T1]{fontenc}
\usepackage[utf8]{inputenc}
\usepackage{authblk}

\usepackage{algorithm}
\usepackage{algorithmic}

\usepackage{hyperref}
\usepackage{multirow}

\newcommand{\N}[2]{\mathcal{N} \left(#1,#2\right) }
\newcommand{\MN}[3]{\mathcal{MN} \left(#1,#2,#3\right) }
\newcommand{\E}[1]{\mathbb{E} \left(#1\right) }
\newcommand{\tr}[1]{\mbox{tr} \left(#1\right) }
\newcommand{\trp}[1]{\mbox{tr}_P \left(#1\right) }
\newcommand{\diag}[1]{\mbox{diag}\left(#1\right)}
\newcommand{\V}[1]{\mathbb{V}\left(#1\right)}

\newtheorem{mylem}{Lemma}

\newcommand*\samethanks[1][\value{footnote}]{\footnotemark[#1]}

\begin{document} 

\title{Network inference in matrix-variate Gaussian models with non-independent noise}
\author[1,2]{Andy Dahl\thanks{These authors contributed equally}}
\author[2]{Victoria Hore\samethanks{}}
\author[2]{Valentina Iotchkova}
\author[1,2]{Jonathan Marchini}

\affil[1]{Wellcome Trust Centre for Human Genetics, Oxford}
\affil[2]{Department of Statistics, Oxford University}

\maketitle
\date{\vspace{-5ex}}
\vskip 0.3in

\begin{abstract}
Inferring a graphical model or network from observational data from a large number of variables is a well studied problem in machine learning and computational statistics. In this paper we consider a version of this problem that is relevant to the analysis of multiple phenotypes collected in genetic studies. In such datasets we expect correlations between phenotypes and between individuals. We model observations as a sum of two matrix normal variates such that the joint covariance function is a sum of Kronecker products. This model, which generalizes the Graphical Lasso, assumes observations are correlated due to known genetic relationships and corrupted with non-independent noise. We have developed a computationally efficient EM algorithm to fit this model. On simulated datasets we illustrate substantially improved performance in network reconstruction by allowing for a general noise distribution.
\end{abstract} 

\section{Introduction}

It is now common for genetic studies of human diseases to collect multiple correlated measurements on individuals to uncover an underlying genetic network. In such studies, $P$ variables (commonly called phenotypes or traits) are measured on $N$ related individuals and stored in a data matrix $Y \in \mathbb{R}^{N\times P}$. 

In this paper we consider a model which is the sum of two matrix normal variates such that the joint covariance function is a sum of Kronecker products:
\begin{eqnarray}\label{eq:model}
Y & = & Z + \epsilon \\
Z & \sim & \MN{0}{R^{-1}}{C^{-1}}\\
\epsilon & \sim & \MN{0}{I}{D^{-1}}
\end{eqnarray}
$Z$ denotes the genetic component of the multiple phenotypes, which we model as a matrix normal distribution with row precision matrix $R$ and column precision matrix $C$. The matrix $R$, which defines the relationships between individuals (e.g. parent-child pairs), will typically be known in advance or be well estimated from genetic data collected on the $N$ individuals \cite{segura2012efficient, kang2008efficient, lippert2011fast}. We therefore assume that it is known in this paper. The term $\epsilon$ denotes the non-genetic component of the multiple phenotypes, which we model as a matrix normal distribution with row precision matrix $I$ and column precision matrix $D$.

In order to infer the genetic network between variables, we assume $C$ is sparse. A sparse precision matrix defines a network between the variables, where non-zero elements correspond to connections in the network. Inference in \eqref{eq:model} is not straight-forward as Kronecker products are not closed under addition, and so $Y$ is not matrix variate normal. Therefore even for small $N$ and $P$, the inference can be prohibitively slow as full covariance matrices of all the random variables in $Y$ contain $N^2P^2$ elements.

Several algorithms already exist for models that are related to \eqref{eq:model}. The simplest special case occurs when $\epsilon=0$. In this case a sparse $C$ can be learned using the Graphical Lasso (Glasso), which puts an $\ell_1$ penalty on $C$ and results in inferences of a Gaussian Graphical Model (GGM) \cite{banerjee2008, yuan2007}.

If, in addition to setting $\epsilon=0$, we allow $R$ to be unknown then the covariance matrices can be learned using a ``flip-flop'' algorithm \cite{leng2012sparse}. This iterative procedure involves rotating the data using one covariance matrix and estimating the other using the resulting whitened data \cite{zhang2010}. However, it was noted by \cite{binkley1988} that matrix variate models can result in information loss through cancellation, which can be avoided by adding $\epsilon$ as in \eqref{eq:model}. 

Another related model is considered in \cite{stegle2011}, in which the observational noise is assumed to be independent and identically distributed (iid) (i.e. $D=\tau I$) but $R$ is learned. Specifically, $R$ is assumed to have an approximately low rank structure so as to model confounding factors. Model inference is performed using an approximate EM algorithm called the Kronecker Glasso (KronGlasso). The noise free random effects, $Z$, are modeled as latent variables and learned alongside the other parameters. Parameters are iteratively updated based on the current estimates of the other parameters in order to optimize the model likelihood. A sparse $C$ is learned using Glasso. The implementation is not an EM algorithm because a fixed estimate of $Z$ is used to estimate $C^{-1}$. 

\cite{kalaitzis2013} recently introduced the Bigraphical Lasso to perform inference in the case where $C^{-1}$ is equal to the identity. With this constraint, the data covariance becomes a Kronecker sum, $R^{-1} \oplus D^{-1}$. As the name suggests, the Bigraphical Lasso alternately performs Glasso on $R$ and $D$ in a flip-flop type algorithm.

The assumption that $R$ is known \textit{a priori} is our focus here, which is a ubiquitous assumption in studies of human, plant and animal genetics \cite{zhao2007arabidopsis, yu2005unified, weir2006genetic}. With $R$ fixed, the KronGlasso algorithm reduces to just three components: estimating the noise precision $\tau$; finding the expectation of $Z$; and estimating a sparse $C$. We feel that the assumption of iid noise is quite restrictive, and for many applications (not just genetics) an arbitrary structure on $D$ would be more appropriate. 

We develop a class of efficient EM algorithms to estimate \eqref{eq:model} under general convex penalty functions for $C$ and $D$. The E-step of our algorithm has computational complexity of $O(NP^2 + P^3)$ at each iteration, and the M-step typically calls efficient third-party software.

In the following section, we describe our full EM algorithm for the case of arbitrary noise. We also explain the difference between our EM algorithm and the approximation of \cite{stegle2011}. Section 3 contains a comparison of the algorithms' ability to infer structure in $C$ on simulated data. We generate data with both iid and non-iid noise, showing that in the case of non-iid noise, allowing for an arbitrary $D$ improves inference of $C$ considerably. A conclusion and discussion is given in section 4. Technical lemmas are given in the appendix.

\section{Methods}

In this section, we derive an EM algorithm to maximize the likelihood of \eqref{eq:model} under arbitrary penalties on $C$ and $D$. Before we describe the algorithm, we layout out some definitions and notation.

\subsection{Definitions and notation}
The Kronecker product of matrices $U$ and $V$ is defined by
\begin{align*}
U \otimes V = 
 \left[
\begin{array}{cccc} 
u_{11} V & u_{12} V & \cdots & u_{1n} V \\
u_{21} V & u_{22} V & \cdots & u_{2n} V \\
\vdots & \vdots & & \vdots \\
u_{n1} V & u_{n2} V & \cdots & u_{nn} V 
\end{array} 
\right]
\end{align*}
We denote $x = \text{vec} (X)$ as the column-wise vectorization of a matrix $X$. If $M$ is an $np \times np$ matrix, we can represent this matrix in terms of $p \times p$ blocks, as
\[M =  \left[ \begin{array}{ccc}
M_{11} & \ldots & M_{1n}\\
\vdots & \ddots & \vdots\\
M_{n1} & \ldots & M_{nn}
\end{array} \right] \]
then define $tr_P(M)$ is the $n \times n$ matrix of traces of such blocks
\[tr_P(M) =  \left[ \begin{array}{ccc}
tr(M_{11}) & \ldots & tr(M_{1n})\\
\vdots & \ddots & \vdots\\
tr(M_{n1}) & \ldots & tr(M_{nn})
\end{array} \right] \]

\noindent
Finally, the matrix variate normal with mean zero has density
\begin{align*}
\MN{Y | 0}{A^{-1}}{B^{-1}}   = \frac{\text{exp}\big{(} -\frac{1}{2} \text{tr} [B Y^T A Y ] \big{)}}{(2 \pi)^{NP/2} |A|^{P/2} |B|^{N/2}} 
\end{align*}
\noindent
\noindent
This is a special case of a multivariate normal, where $\text{vec}(Y)$ has mean 0 and precision $B \otimes A$. 

\subsection{A penalized EM algorithm}

The EM algorithm consists of an E-step which calculates an objective function and an M-step to maximise this objective function. Treating $Z$ as a latent variable in \eqref{eq:model}, the objective function at step $t$ is given by
\begin{align}
Q\left( C, D | C^{(t)}, D^{(t)}, R \right) &= \mathbb{E}_{Z| \Theta^{(t)}} \left[ \log P\left( Y, Z | C, D, R\right) \right] \notag \\
 &= \mathbb{E}_{Z| \Theta^{(t)}} \left[\log P( Y | D, Z ) + \log P( Z | R, C ) \right] \label{eq:asym}
\end{align}
where we denote $\mathbb{E}_{Z|Y, R, C^{(t)}, D^{(t)}} := \mathbb{E}_{Z| \Theta^{(t)}}$. The individual terms in the above expression can be re-written as
\begin{align}
 \mathbb{E}_{Z| \Theta^{(t)}} \left[\log P( Y | D, Z ) \right] &\equiv \mathbb{E}_{Z| \Theta^{(t)}} \big[ N \log | D | - || ( D^{1/2} \otimes I_N) (y-z) ||_2^2  \big] \nonumber \\
&\equiv N \log | D | -  \mathbb{E}_{Z| \Theta^{(t)}} \left[  \tr{ (Y-Z) D (Y-Z)^T} \right]  \nonumber \\
&\equiv N \log | D | - N \tr{ D \Omega_1^{(t)} } \label{exp1} \\
\nonumber \\
 \mathbb{E}_{Z| \Theta^{(t)}} \left[ \log P( Z | R, C ) \right] &\equiv \mathbb{E}_{Z| \Theta^{(t)}} \big[  N \log |C| - || \left( C^{1/2} \otimes R \right) z ||_2^2   \big] \nonumber \\
&\equiv N \log |C| - \mathbb{E}_{Z| \Theta^{(t)}} \tr{ R Z C Z^T} \nonumber \\
&= N \log |C| - N \tr{ C \Omega_2^{(t)} }  \label{exp2} 
\end{align}
where $\Omega_1^{(t)}$ and $\Omega_2^{(t)}$ are estimates of the unobserved sample covariance matrices $Z^T Z$ and $\epsilon^T \epsilon$,
\begin{eqnarray}
\Omega_1^{(t)}	&:=& \mathbb{E}_{Z| \Theta^{(t)}} \left[ \frac{1}{N}  (Y-Z)^T (Y-Z) \right] \label{def:Omega-1} \\
\Omega_2^{(t)}	&:=& \mathbb{E}_{Z| \Theta^{(t)}} \left[ \frac{1}{N}  Z^T R Z  \right] \label{def:Omega-2}
\end{eqnarray}

\noindent
If $Z$ were known, the terms inside these expectations would be the obvious estimators for $C$ and $D$. 

Together, \eqref{exp1} and \eqref{exp2} imply that maximizing $Q$ is equivalent to minimizing
\begin{eqnarray}
-\log | D | - \log |C| + \tr{ D \Omega_1^{(t)} } + \tr{ C \Omega_2^{(t)} }  \label{eq:Q}
\end{eqnarray}

\noindent
We note that even though the function \eqref{eq:asym} treats $Z$ and $\epsilon$ asymmetrically, the symmetry is recovered in the EM objective function \eqref{eq:Q}. This symmetry is not recovered in the KronGlasso algorithm derived in \cite{stegle2011}.

The M-step optimizes the objective \eqref{eq:Q} with added penalty $\mathcal{P} \left(C,D\right)$. If this penalty additively separates into convex functions of $C$ and $D$, so that $\mathcal{P} \left(C,D\right) = \mathcal{P}_C \left(C\right)  + \mathcal{P}_D \left(D\right)$, the M-step becomes two uncoupled convex optimization problems:
\begin{align*}
D^{(t+1)} \leftarrow & \min_{ D \succ 0 } \left( - \log | D | + \tr{ D \Omega^{(t)}_1 }   + \mathcal{P}_D \left(D\right)  \right) \\
C^{(t+1)} \leftarrow & \min_{C \succ 0 } \left( - \log |C| + \tr{ C \Omega^{(t)}_2 } + \mathcal{P}_C \left(C\right) \right)
\end{align*}

\noindent
We primarily use $\mathcal{P} \left(C,D\right) = \lambda ||C||_1$, corresponding to the belief that $D$ is dense and that $C$ describes a sparse graphical model. The resulting EM algorithm is 
\begin{enumerate} 
\item Compute $\Omega_1^{(t)}$ and $\Omega_2^{(t)}$
\item Update $D_{t+1}	\leftarrow \Omega_1^{(t)} $; $C_{t+1}	\leftarrow \text{Glasso} \left( \Omega_2^{(t)}, \lambda \right)$
\end{enumerate}
We refer to our method as G$^3$M since it infers a Genetic Gaussian Graphical Model.
 
\subsection{Evaluating the $\Omega$'s}

To compute the expectations in (\ref{def:Omega-1}) and (\ref{def:Omega-2}), we require the full conditional distribution for $Z$, given by
\begin{align*}
P( Z |Y,C,D)	&\propto P(Y|Z,D) P(Z|C) \\
	&\propto \exp \frac{-1}{2} \left( (y-z)^T \left[ D \otimes I \right] (y-z) + z^T \left[C \otimes R \right] z \right) \\
	&\propto \exp \frac{-1}{2} \left( -2z^T \left[ D \otimes I \right] y + z^T \left[ D \otimes I + C \otimes R \right] z \right) \\
	&= \exp \frac{-1}{2} \left( -2z^T \Sigma^{-1} \left( \Sigma \left[ D \otimes I \right] y \right) + z^T \Sigma^{-1} z \right) \implies \\
	\vspace{3mm}
z &  |Y,C,D	\sim \N{\mu}{\Sigma }
\end{align*}
where
\begin{eqnarray}
\Sigma	&:=& \left[ D \otimes I + C \otimes R \right]^{-1} \label{def:Sigma} \\
\mu	&:=& \Sigma \left[ D \otimes I \right] y \label{def:mu}
\end{eqnarray}
\noindent
Define $M := \mbox{vec}^{-1} \left( \mu \right)$, that is, fill up an $N \times P$ matrix column-wise with the entries of $\mu$. Then the $\Omega$'s can be rewritten as
\begin{eqnarray}
\Omega_1^{(t)}	&= & \frac{1}{N} \left[ (Y-M_t)^T(Y-M_t) + \trp{\Sigma_t} \right]  \label{s1}\\
\Omega_2^{(t)}	&= & \frac{1}{N} \left[ M_t^T R M_t + \trp{\left( I_P \otimes R \right) \Sigma_t} \right] \label{s2} 
\end{eqnarray}
using the result 
\begin{align}
\mathbb{E} \left(  X^T R X \right)_{ij}	&= \E{ \tr{ X_{,j} R X_{,i}^T } } = \tr{ R \left( \nu_{,j}  \nu_{,i}^T + \mbox{Cov}( X_{,i}, X_{,j} ) \right) }  \notag	\\
\implies & \mathbb{E} \left[ X^T R X \right]	= \nu^T R \nu + \trp{\left( I \otimes R \right) \Theta}	\notag 
\end{align}
for any $X$ such that $\E{X} = \nu$ and $\V{ \mbox{vec}(X) } = \Theta$. 

\eqref{s1} and \eqref{s2} give explicit forms for the E-step and are only $O(NP^2)$. However, before these computations can be performed, $M$ and $\Sigma$ must be computed, which costs $O(N^3P^3)$ as written in \eqref{def:Sigma} and \eqref{def:mu}. We use simple linear algebra tricks in the next section to decrease the complexity to $O(NP^2 + P^3)$. 

\subsubsection*{Efficient computation}

We begin with computing the following (where we have dropped reference to the iteration $(t)$ for clarity),
\begin{align*}
U \Lambda_R U^T		:&= \mbox{Spectral Decomposition} \left( R \right) \\
Q_1 \Lambda_1 Q_1^T	:&= \mbox{Spectral Decomposition} \left( D^{-1/2} C D^{-1/2} \right) \\
\Lambda_1^*			:&= \left[ I + \Lambda_1 \otimes \Lambda_R \right]^{-1} \\
Q_2 \Lambda_2 Q_2^T	:&= \mbox{Spectral Decomposition} \left( C^{-1/2} D C^{-1/2} \right) \\
\Lambda^*_2			:&= \left[ I + \Lambda_2 \otimes \Lambda_R^{-1} \right]^{-1}
\end{align*}
These computations require that $C$, $D$ and $R$ are invertible, however we note that this is guaranteed by the log-determinant terms in the likelihood function. Also compute
\begin{align*}
S_1	&= \mbox{vec}^{-1} \left( \diag{ \Lambda_1^* } \right) \ast \left( U^T Y D^{1/2} Q_1 \right) \\
S_2	&= \mbox{vec}^{-1} \left( \diag{ \Lambda_2^* } \right) \ast \left( U^T Y C^{1/2} Q_2 \right) 
\end{align*}
In the remainder of this section, we show how to compute the $\Omega$'s using only the quantities above. The algebraic results that we use are given in the Appendix. 

The $\Sigma$ terms in \eqref{s1} and \eqref{s2} are easy to write in terms of the above quantities:
\begin{align*}
\mbox{tr}_P \left( \Sigma 	\right)								&= \mbox{tr}_P \left( \left[ D \otimes I + C \otimes R \right]^{-1} \right)	\notag	\\
	&= D^{-1/2} Q_1 \mbox{tr}_P \left( \Lambda_1^* \right) Q_1^T D^{-1/2} \label{eq:sig1} \tag{by Lemma 3} 	\\
\mbox{tr}_P \left( \Sigma \left( I \otimes R \right) \right)	&= \mbox{tr}_P \left( \left[ D \otimes R^{-1} + C \otimes I \right]^{-1} \right)	\notag \\
	&= C^{-1/2} Q_2  \mbox{tr}_P \left( \Lambda^*_2 \right) Q_2^T C^{-1/2} \tag{by Lemma 3}
\end{align*}

\noindent
The terms involving $M_t$ in \eqref{s1} and \eqref{s2} require a bit more work:
\begin{align}
(Y-&M_t)^T (Y-M_t)	= \mbox{tr}_P \left( (y-\mu_t) (y-\mu_t)^T \right)		\tag{by Lemma 4}	\\
	&\,{\buildrel (a) \over =}\, \mbox{tr}_P \Big[ \left( \left[ C^{-1}D \right] \otimes R^{-1} + I \right)^{-1} y  y^T \left( \left[ C^{-1}D \right] \otimes R^{-1} + I \right)^{-T} \Big] \notag \\
	&\,{\buildrel (\dagger) \over =}\,  \left( C^{-1/2} Q_2 S_2^T \right) \left( C^{-1/2} Q_2 S_2^T \right)^T	\label{eq:M-1} \notag
\end{align}
\begin{align}
M_t^T & R M_t		= \mbox{tr}_P \left( \mu_t \mu_t^T \left( I \otimes R \right)	 \right) \tag{by Lemma 4} \notag \\	
	&\,{\buildrel (b) \over =}\, \mbox{tr}_P \Big[ \left( \left[ D^{-1}C \right] \otimes R + I \right)^{-1} y  y^T \left( \left[ D^{-1}C \right] \otimes R + I \right)^{-T}   \left( I \otimes R \right) \Big] \notag \\
	&\,{\buildrel (\dagger) \over =}\, \left( D^{-1/2} Q_1 S_1^T \right) \Lambda_R \left( D^{-1/2} Q_1 S_1^T \right)^T	\label{eq:M-2}\notag
\end{align}
\noindent
(a) and (b) must be proven, while the $(\dagger)$ equations use Lemma 5. For (a),
\begin{align*}
I - \Sigma \left(D \otimes I \right)	&= I - \left( D \otimes I + C \otimes R \right)^{-1} \left( D \otimes I \right)	\\
	&= I - \left( I + (D^{-1}C) \otimes R \right)^{-1} \\
	&\,{\buildrel (*) \over =}\, \left( (C^{-1}D) \otimes R^{-1} + I \right)^{-1} \implies \\
(y-\mu)	&= \left[ I - \Sigma \left(D \otimes I \right) \right] y = \left( (C^{-1}D) \otimes R^{-1} + I \right)^{-1} y
\end{align*}
and for (b),
\begin{eqnarray*}
\mu	= \Sigma \left( D \otimes I \right) y &=& \left( C \otimes R + D \otimes I \right)^{-1} \left( D \otimes I \right) y\\
 &=& \left( \left[ D^{-1}C \right] \otimes R + I \right)^{-1} y 
\end{eqnarray*}
 
\noindent
Finally, equation $(*)$ follows from
\[ I - (I+P)^{-1} = (I+P)^{-1} P = (P^{-1} + I )^{-1} \]

\subsubsection*{Runtime and memory complexity}
 
Our manipulation of the expressions for the $\Omega$'s using eigenvalue decompositions and standard results of Kronecker products reduces these calculations to $O(N^3 + P^3)$. This can be further improved by performing a one-off eigendecomposition of $R$ at the outset, which reduces the complexity of computing the $\Omega$'s to $O(N P^2 + P^3)$. For many datasets, this will mean that the run time of our method is dominated by the optimization in the M-step. For example, Glasso has complexity $O(P^4)$, and even though this code is extremely well optimized our algorithm spends the vast majority of its time running Glasso in all our simulations.


\subsection{Relationship to the KronGlasso }

The KronGlasso algorithm of \cite{stegle2011} is only an approximate EM algorithm in that it does not use the full conditional distribution of $Z$. Rather, at each step it calculates the expectation of $Z$ and uses this to estimate $C$. In the setting where $R$ is known, the expression for $\Omega_2^{(t)}$ reduces to 
\begin{eqnarray}
 \Omega_2^{(t)} = \frac{1}{N} \left( Z^T R Z  \right) \notag 
\end{eqnarray}

Another difference between the algorithms is in estimating $D$. \cite{stegle2011} assumes $D$ is iid and learns the scalar variance parameter by gradient ascent. Although our framework allows a general $D$, a penalty can be used to constrain $D$ to be iid. Another benefit of our exact approach is that the variance parameter has an analytic solution: the update for $D=\tau I$ is
\begin{eqnarray}
 \tau^{(t+1)} \leftarrow \min_{ \tau \geq 0 } \left( - P\log ( \tau ) + \tau \tr{  \Omega^{(t)}_1 }   \right)  =  \frac{P}{\tr{\Omega_1^{(t)}}} \notag
\end{eqnarray}

\section{Simulation study}

\subsection{Data Generation}

We carried out a simulation study to illustrate the benefits of modelling non-independent noise for graphical model estimation. We simulated 40 datasets with $N=400$ individuals and $P=50$ traits according to model \eqref{eq:model}. We assumed a relatedness matrix $R$ with a block diagonal structure of 80 families of 5 siblings, so that each block of 5 individuals has off-diagonal entries equal to 0.5.

We vary $C$ and $D$ in our simulations to demonstrate different levels of sparsity, summarized in Table 1. We generate a matrix defined by Random($p$) in two steps. First a fraction $p$ of the edges are taken to be non-zero and equal. Second we add a scalar multiple of the identity such that the resulting condition number of the matrix is exactly $P$, as in \cite{leng2012sparse, kalaitzis2013}. For AR(1) we use an autocorrelation of 0.8.

We scaled $C$ and $D$ so that the signal-to-noise ratio (SNR) was 20\%, where we define 
\begin{eqnarray}
 \text{SNR} = \frac{\mathbb{E}||Z||_F^2}{\mathbb{E}||\epsilon||_{F}^2} = \frac{tr(C^{-1})}{tr(D^{-1})} \notag
\end{eqnarray}
In the genetics literature, the size of the genetic signal is measured on a trait by trait basis as the heritability $(h^2_i)$ of the $i$th trait, which is the proportion of total variance attributable to the genetic random effect. We generate equally heritable traits, so that
\begin{eqnarray}
 h^2_i = \frac{(C^{-1})_{ii}}{(C^{-1})_{ii}+(D^{-1})_{ii}} = \frac{ \tr{ C^{-1} } } { \tr{ C^{-1} }+\tr{D^{-1}} } = \frac{SNR}{1+SNR} \notag
\end{eqnarray}
Thus the heritability of each trait is constant and 0.17 in all our simulations.

\begin{table}
\begin{center}
\begin{tabular}{|c|c|} \hline
$C$				& $D$ \\ \hline
AR(1)	& AR(1) \\ \hline
Random(1\%)		& Wishart $\big{(}P\!-\!3, I_P /$ \small{$(P\!-\!3)$}$\big{)}$ \\ \hline
Random(10\%)	&  iid \\ \hline
\end{tabular}
\end{center}
\label{tab:C-D}
\caption{ Different choices for precision matrices. }
\end{table}

\subsection{ROC Curves}

We compare G$^3$M to KronGlasso and vanilla Glasso. Vanilla Glasso is run on the sample covariance $Y^T Y$ and we implemented KronGlasso ourselves based on the details in \cite{stegle2011}. We use the R package \texttt{glasso} by \cite{Friedman2008} whenever we call Glasso. To make a fair comparison with G$^3$M, which knows the true $R$, we give KronGlasso the true $R$. 

Each of these methods has a regularization parameter $\lambda$ which we vary by setting $\lambda=5^x$ for $x$ linearly interpolated between -7 and 3. For each value of $\lambda$ we infer a network and calculate its power and type I error for picking edges in $C$. Figure 1 presents ROC curves for Glasso, KronGlasso and G$^3$M, each averaged over all 40 datasets. Each plot corresponds to a unique pair of $C$ and $D$; rows and columns index $C$ and $D$ respectively. 

In the case of a dense $D$ (column 1), the motivating scenario for our method, we find that G$^3$M performs considerably better than both Glasso and KronGlasso. Our improvement is uniform; for all type 1 error levels, our method offers the greatest power. 

When $D$ is iid (column 3), our method performs worst. It is not surprising that G$^3$M loses, as the simplifying assumption of iid noise made by KronGlasso is satisfied. However, it is surprising that KronGlasso loses to vanilla Glasso. We suspect that for our choice of $R$ the parsimony of Glasso outweighs the flexibility of KronGlasso; the improvements over vanilla Glasso noted in \cite{stegle2011} used an approximately low rank $R$, which is more structured than our $R$ matrix and understandably would favor KronGlasso over vanilla Glasso.

Finally, when $D$ itself is sparse there is no uniformly dominant method. Again, this is unsuprising; G$^3$M expects a dense noise matrix while KronGlasso expects iid noise, and sparse $D$ matrices lie between these extremes. We do note, however, that G$^3$M performs drastically better for type 1 error rates less than 10\%, which is the interesting part of the ROC curve for almost all statistical applications. 

Although we expect dense noise matrices in practice, which motivated our choice not to penalize $D$, it is easy to adapt G$^3$M to model sparse noise if desired. To test this, we add the penalty $\gamma || D ||_1$ to the likelihood. We allow $ \gamma \neq \lambda$, which comes at the cost of searching over a 2 dimensional grid of regularization parameters. We superimposed the resulting ROC curve in green in the center plot of Figure 1, which uses sparse $C$ and $D$. Specifically, we fit the model for all $(\gamma,\lambda)$ pairs, optimize over $\gamma$ for each $\lambda$ and then plot the ROC curve as $\lambda$ varies. Because of this computational cost, we use two shortcuts. First, rather than optimizing $\gamma$ via cross validation, we choose it to maximize the precision of the resulting $C$, which is only possible because we know the ground truth. Second, the line is not averaged over 40 datasets but rather only one. Despite these caveats, we feel it is clear that this regularization on $D$ 
recovers most of G$^3$M's suboptimality in the case of sparse $D$, as expected.
%
%
%
%
%
%
%
%
%
%
%

\begin{figure} \label{ROCs}

\begin{flushleft}
\begin{tabular}{|c|c|ccc|} \hline

& \multicolumn{4}{c}{ $\rightarrow$ Density $\rightarrow$ } \vline \\ \hline

& & Wishart $D$	& AR(1) $D$	& Spherical $D$	\\ \hline

\parbox[t]{2mm}{\multirow{3}{*}{\rotatebox[origin=c]{90}{$\leftarrow$ Density $\leftarrow$ }}} 

& \rotatebox{90}{\ \ \ \ \ \ \ Random(1\%) $C$ }	&
\includegraphics[width=.3\linewidth]{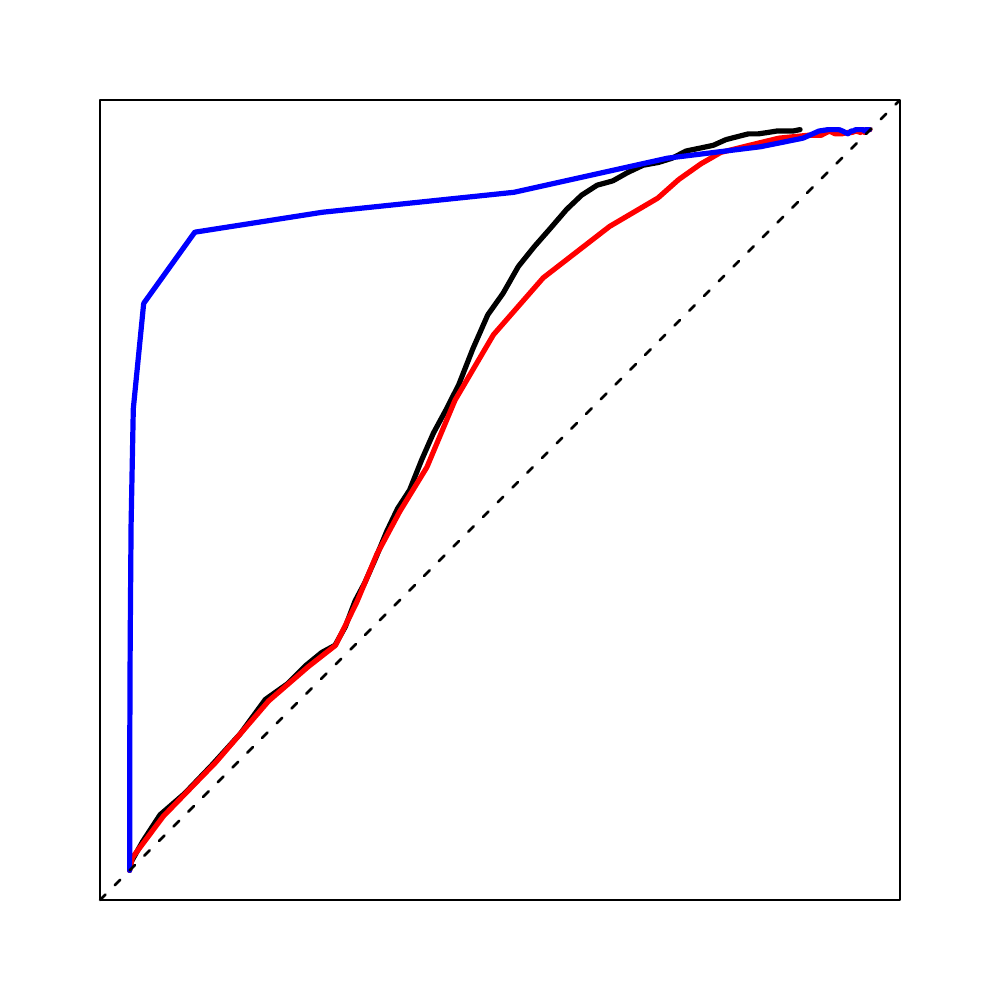} &
\includegraphics[width=.3\linewidth]{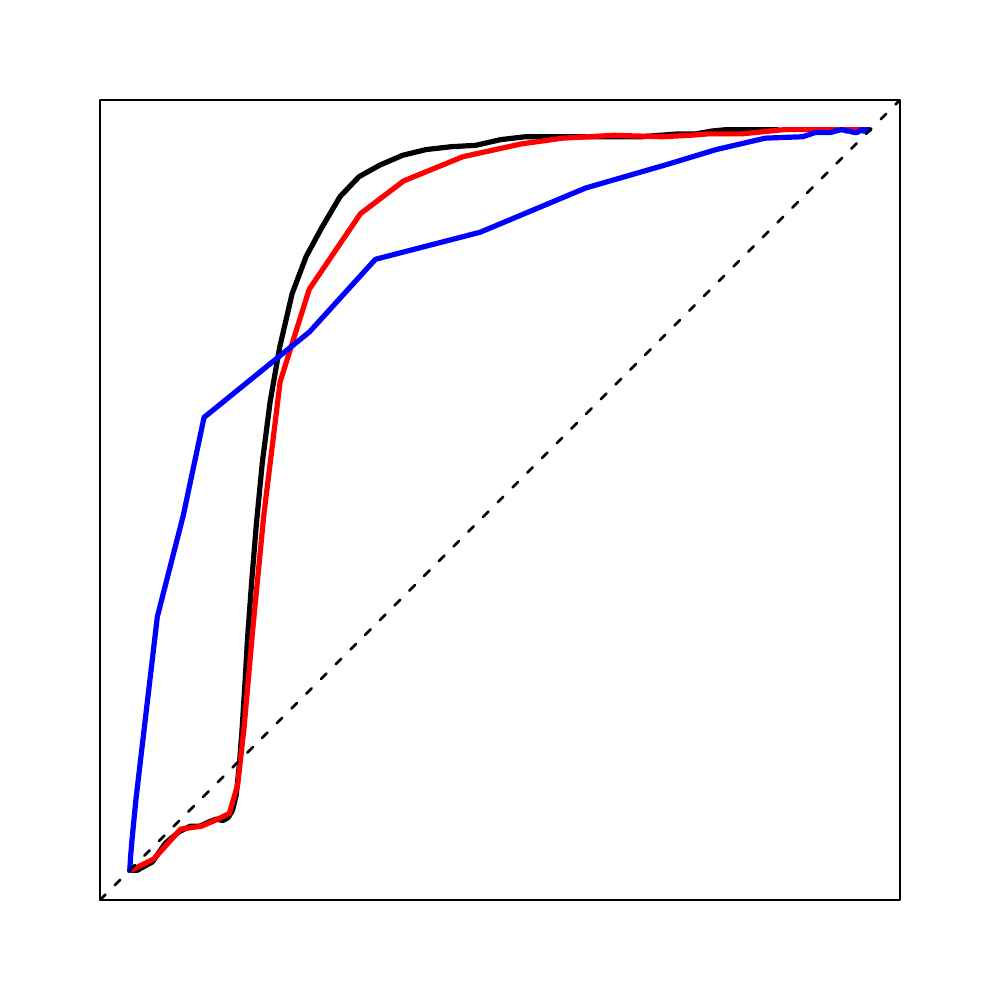} &
\includegraphics[width=.3\linewidth]{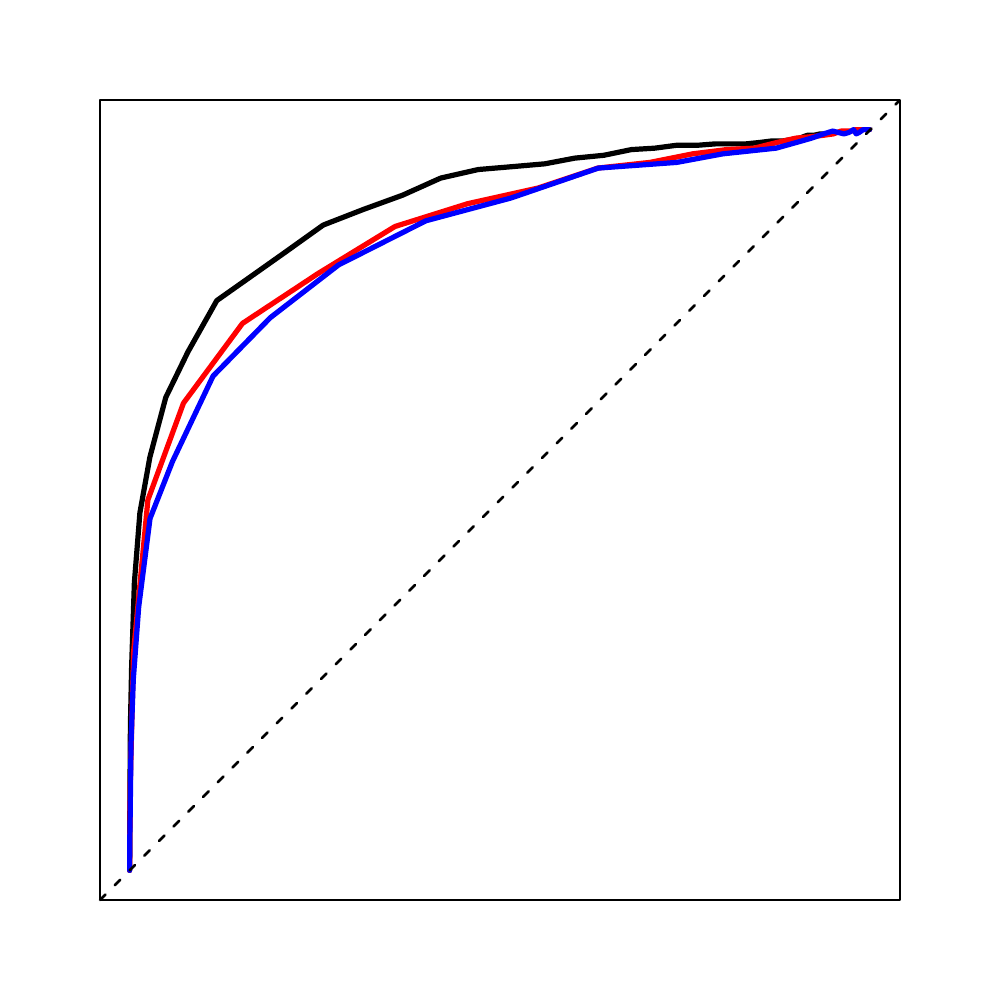} \\

& \rotatebox{90}{\hspace{1.2cm} AR(1) $C$ }	&
\includegraphics[width=.3\linewidth]{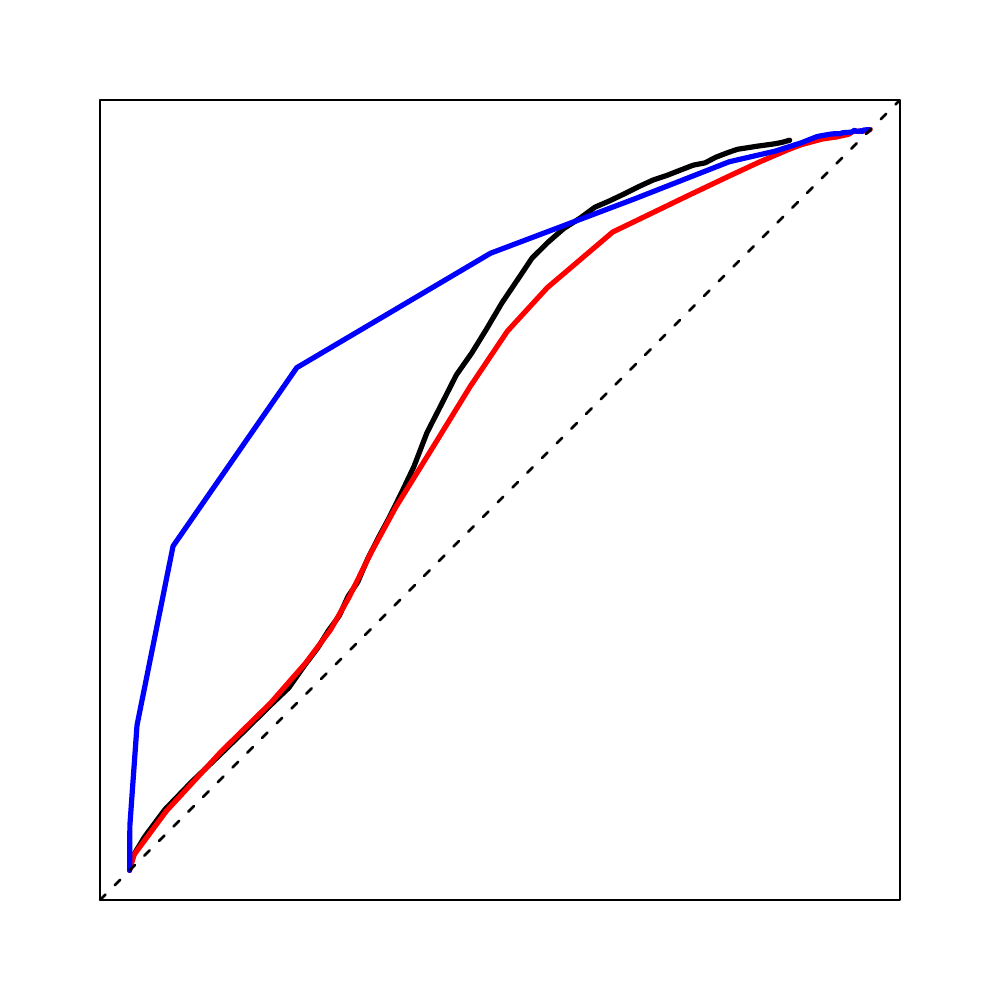} &
\includegraphics[width=.3\linewidth]{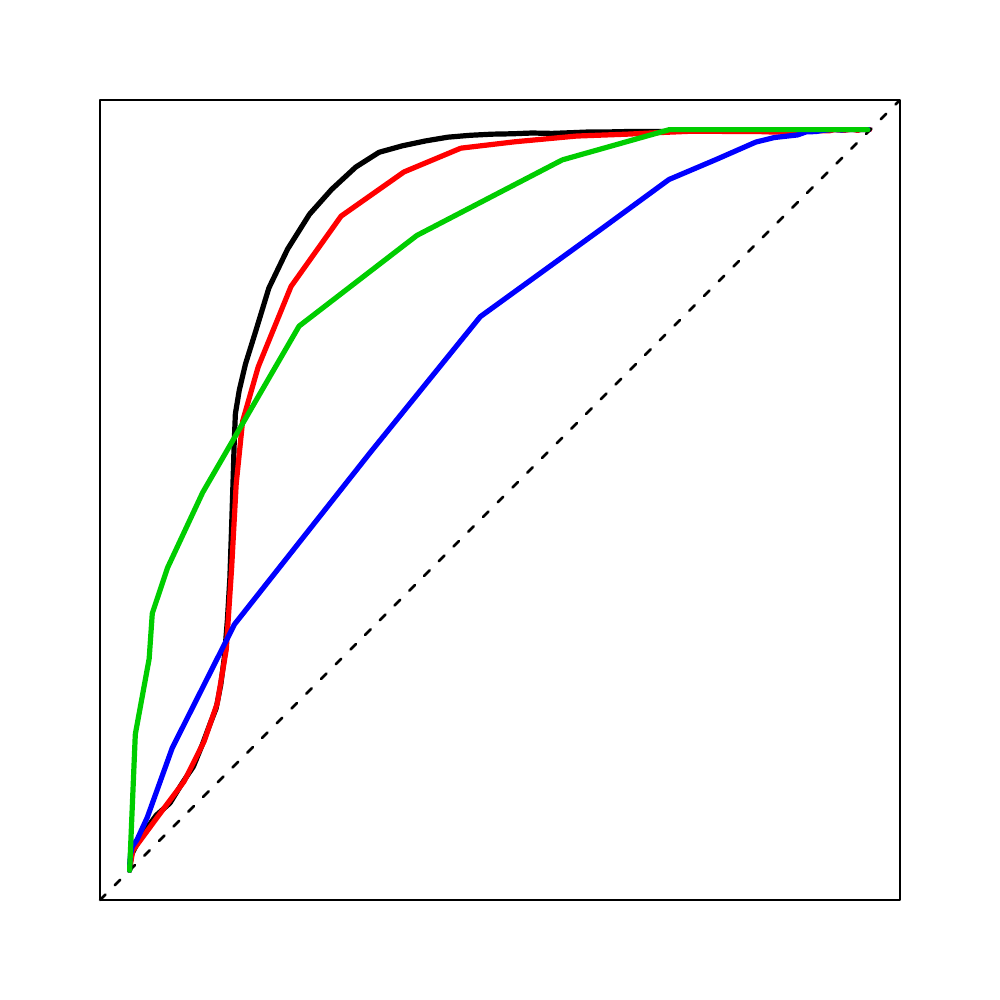} &
\includegraphics[width=.3\linewidth]{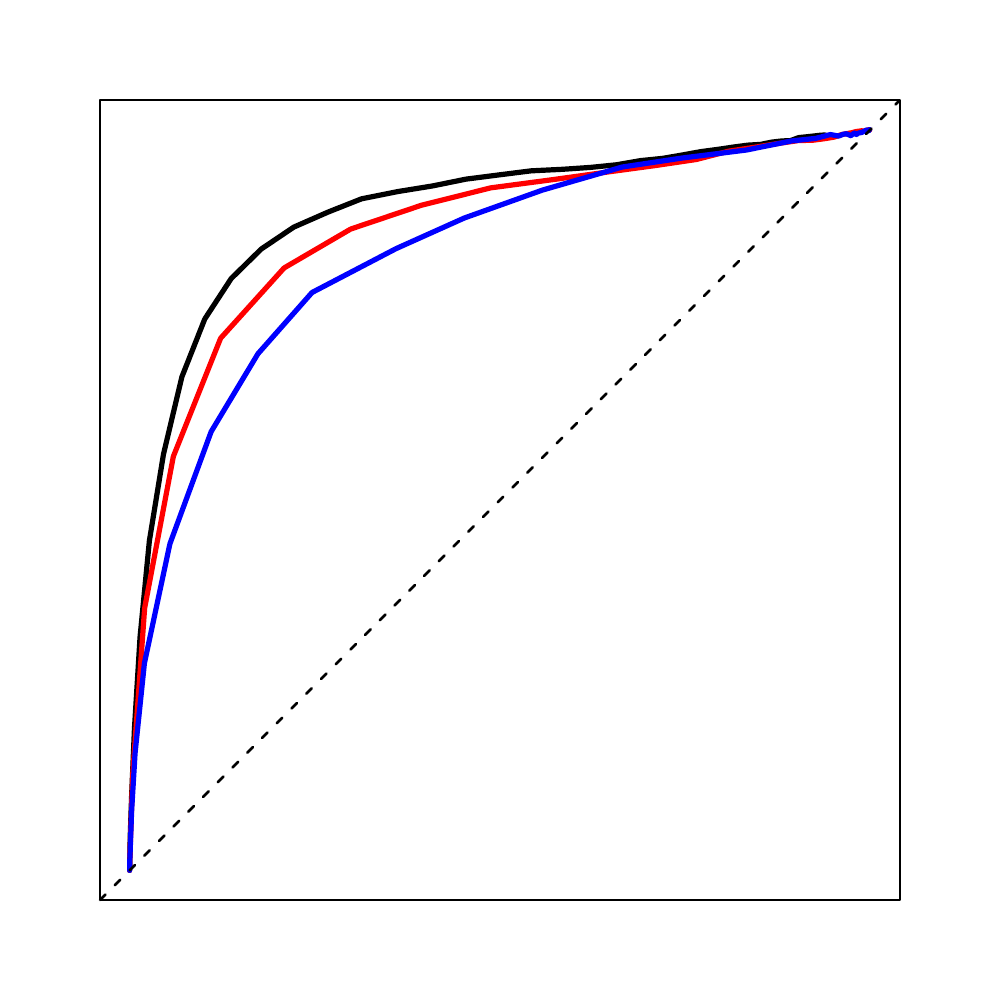} \\

& \rotatebox{90}{\hspace{0.65cm} Random(10\%) $C$ }	&
\includegraphics[width=.3\linewidth]{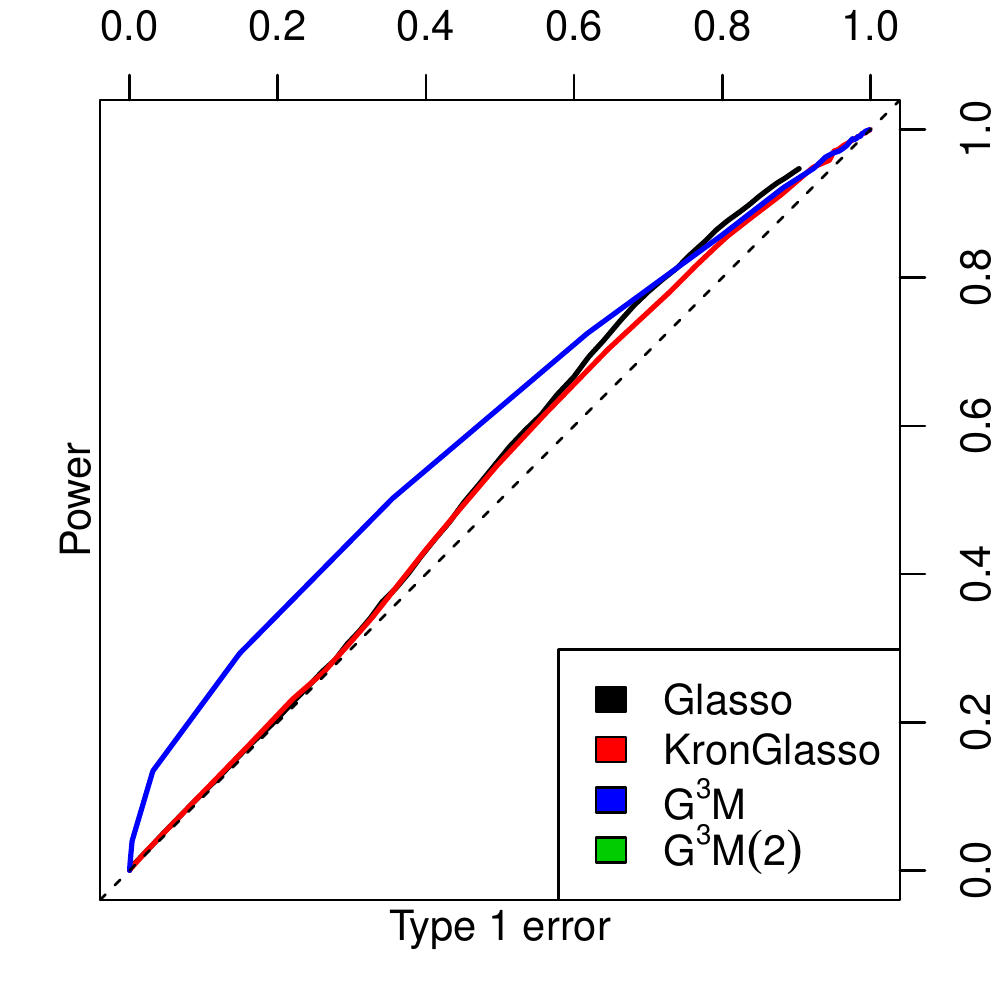} &
\includegraphics[width=.3\linewidth]{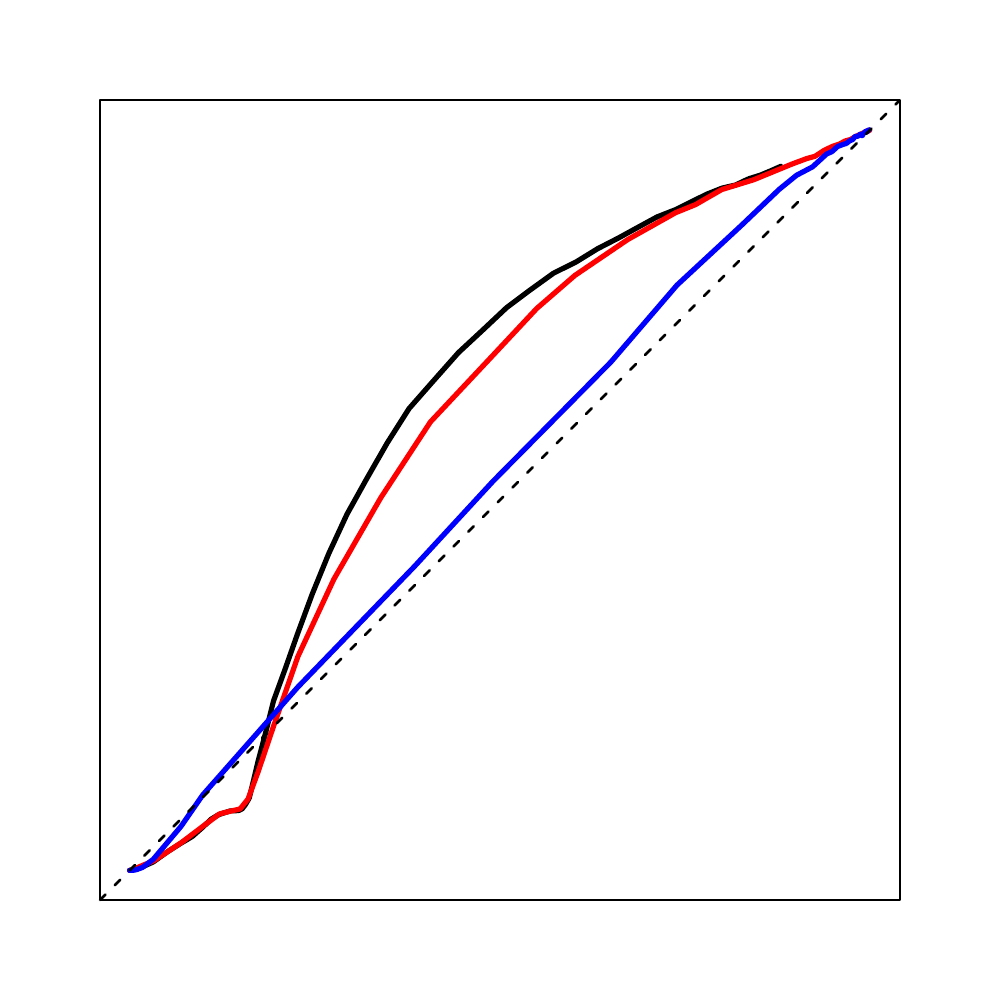} &
\includegraphics[width=.3\linewidth]{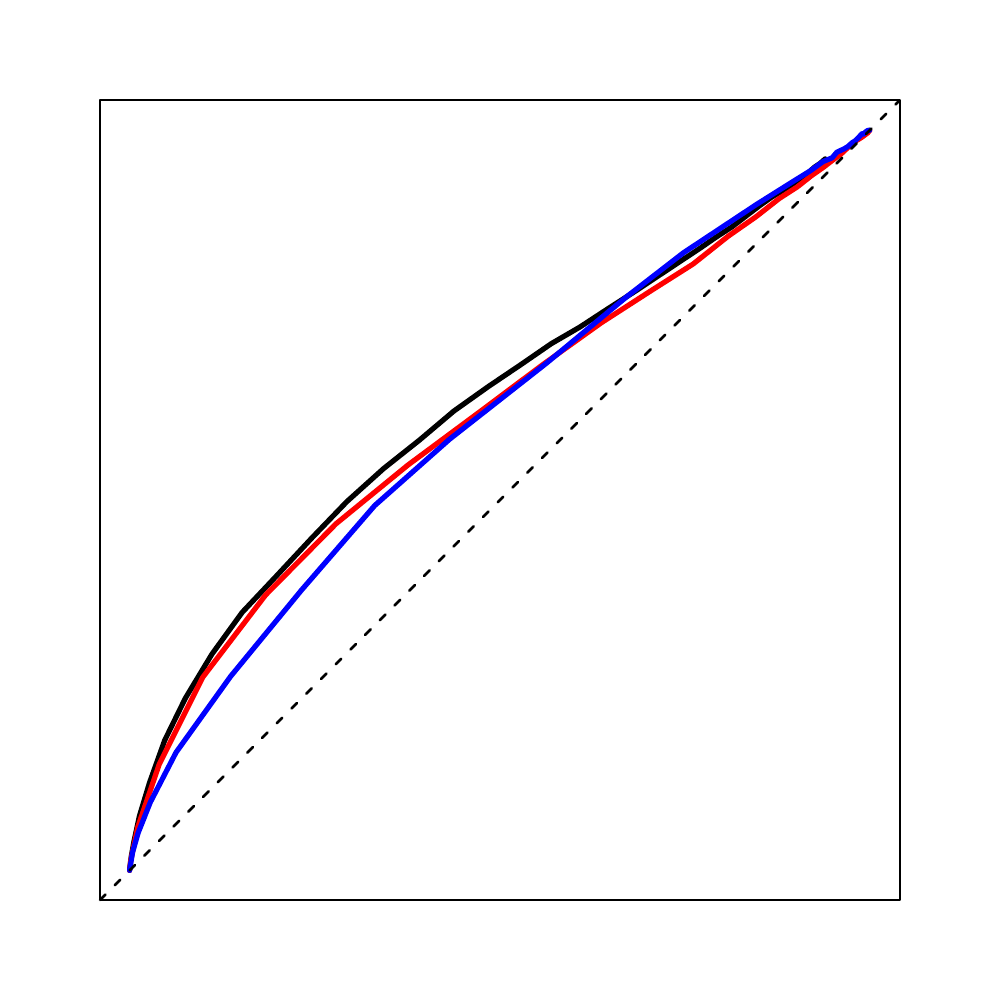} \\ \hline

\end{tabular}
\end{flushleft}

\caption{Comparison of methods for network reconstruction with different $C$ and $D$.}

\end{figure}

\subsection{Network reconstructions}

Typically, one graphical model is selected to summarize the data, and so we compare individual inferred networks from each method on a single simulated data set. This data set was generated using the Random(1\%) model for $C$ and a (dense) Wishart matrix for $D$, which we feel is realistic. Figure 2 shows each method's reconstructed network at 70\% power (shown by the dashed horizontal line). Glasso and KronGlasso return networks that are both unusable practically and give the false impression that the variables are densely related. G$^3$M, however, recapitulates most of the true relationships without including a qualitatively misleading number of false edges. Figure 2 also includes the ROC curve for the specific data set used to generate these networks which is typical of ROC curves averaged to give the (2,2) block of Figure 1.

\begin{figure} \label{simulated-network}
\begin{center}
\begin{minipage}{.45 \linewidth}
\includegraphics[width=\linewidth]{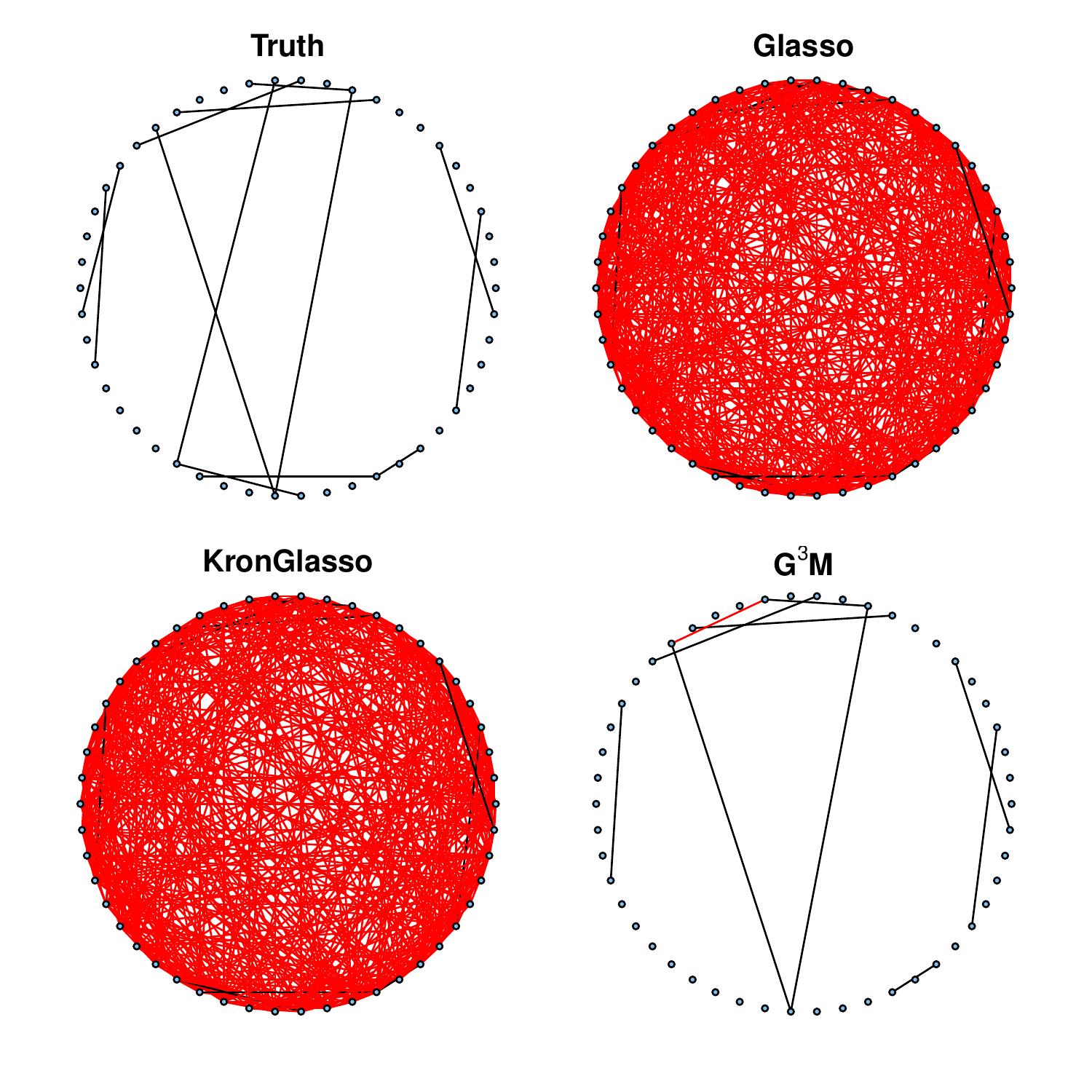}
\end{minipage}
\begin{minipage}{.45 \linewidth}
\includegraphics[width=\linewidth]{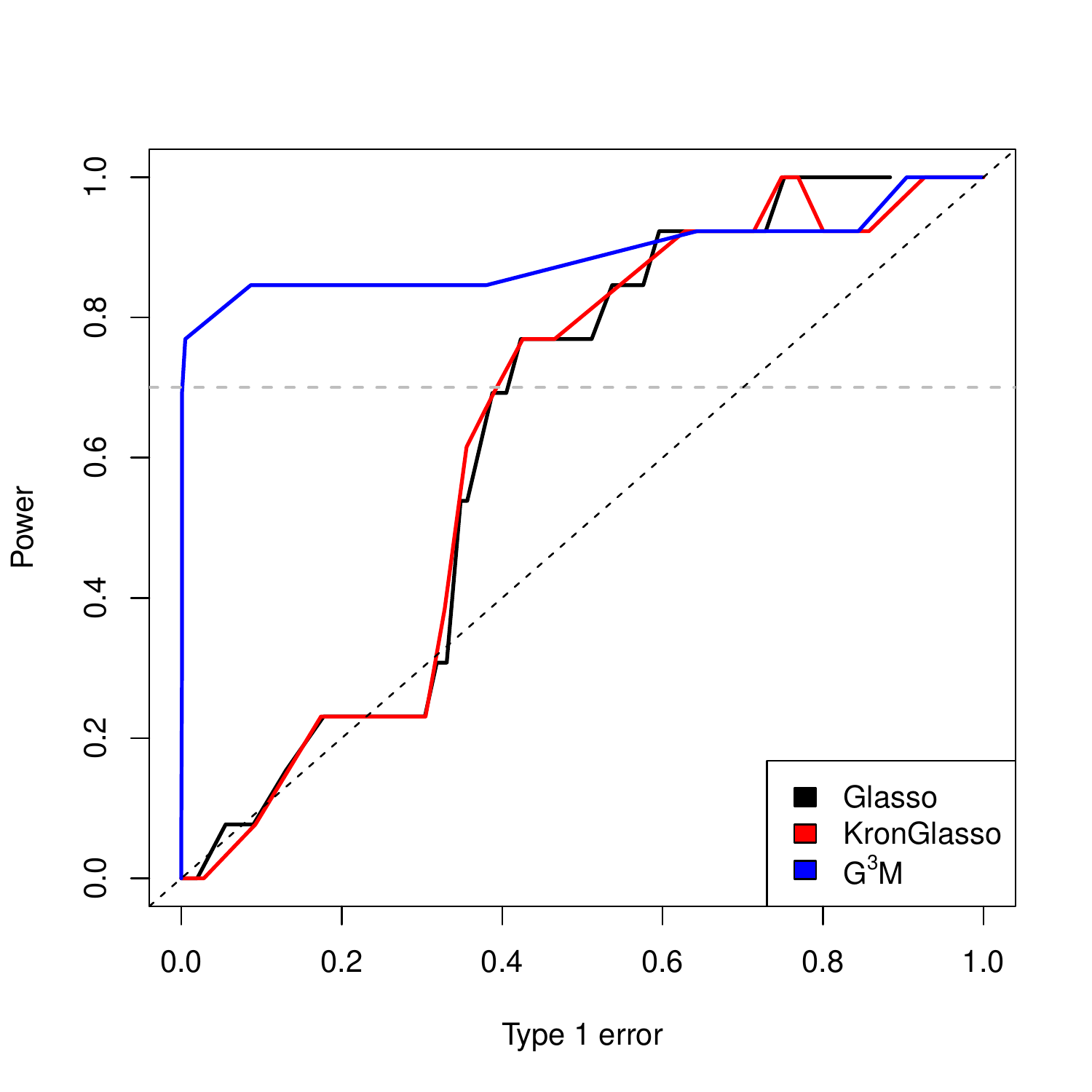}
\end{minipage}
\end{center}

\caption{Network estimation on simulated datasets with Wishart noise. Left : reconstructed networks at 70\% power. Right: ROC curve for the networks shown on the left. The dashed line is drawn at 70\% power.}
\end{figure}

\section{Discussion}

We have developed an efficient EM algorithm for estimation of a graphical model from observational data assuming that samples are correlated and observations are corrupted with non-independent noise. We assume that observations (rows of the data matrix) are correlated with a known covariance structure, a condition that is met by modern multi-phenotype genetic datasets. 

Our approach extends the KronGlasso method of \cite{stegle2011} to accommodate non-iid noise. Moreover, whereas KronGlasso is an approximate EM algorithm, we derive a full EM algorithm and use linear algebra tricks to facilitate computationally efficient inference. Results on simulated datasets show substantial benefits in modelling non-independent noise when it is present. We advocate model selection procedures, such as BIC, out-of-sample predictive error or out-of-sample likelihood, to decide whether iid or non-iid noise models are more appropriate. 

We expect that the EM algorithm can be extended to learn $R$ as in \cite{stegle2011}. At a minimum, the cost would be that the EM algorithm turns into an expectation conditional maximization algorithm, as it is unlikely that the updates for $C$ and $R$ will be easily decoupled.

In future work it will be interesting to explore other penalty functions on $C$ and $D$. A penalty function like
\begin{eqnarray}
 \mathcal{P} \left(C,D\right) = \min_{L+S=C} \left( \lambda_C^S ||S||_1 + \lambda_C^L ||L||_\ast \right) + \lambda_D ||D||_1
 \notag
\end{eqnarray}
\noindent
would model $C$ as a combination of a low-rank component $L$, corresponding to confounders, and a sparse graphical model $S$, encoding causal structure \cite{chandrasekaran2010}. This could be particularly useful in genetics, where the random effect $Z$ is sometimes used to model genome-wide causal effects (which presumably would correspond to a sparse graphical model, suggesting the penalty $||C||_1$) and $Z$ is sometimes used to model confounding population structure (which corresponds to a low-rank confounder, suggesting the penalty $||C||_\ast$). 

We have not explored the utility of this model for prediction of missing phenotypes but this will likely be an important application in real genetic studies where some phenotypes are not measured on all subjects. 


\bibliographystyle{plain}
\bibliography{library}

\begin{thebibliography}{10}

\bibitem{banerjee2008}
Onureena Banerjee, Laurent El~Ghaoui, and Alexandre d'Aspremont.
\newblock Model selection through sparse maximum likelihood estimation for
  multivariate gaussian or binary data.
\newblock {\em The Journal of Machine Learning Research}, 9:485--516, 2008.

\bibitem{binkley1988}
James~K Binkley and Carl~H Nelson.
\newblock A note on the efficiency of seemingly unrelated regression.
\newblock {\em The American Statistician}, 42(2):137--139, 1988.

\bibitem{chandrasekaran2010}
Venkat Chandrasekaran, Pablo~A Parrilo, and Alan~S Willsky.
\newblock Latent variable graphical model selection via convex optimization.
\newblock In {\em Communication, Control, and Computing (Allerton), 2010 48th
  Annual Allerton Conference on}, pages 1610--1613. IEEE, 2010.

\bibitem{Friedman2008}
Jerome Friedman, Trevor Hastie, and Robert Tibshirani.
\newblock Sparse inverse covariance estimation with the graphical lasso.
\newblock {\em Biostatistics}, 9(3):432--441, 2008.

\bibitem{kalaitzis2013}
Alfredo Kalaitzis, John Lafferty, Neil Lawrence, and Shuheng Zhou.
\newblock The bigraphical lasso.
\newblock In {\em Proceedings of the 30th International Conference on Machine
  Learning (ICML-13)}, pages 1229--1237, 2013.

\bibitem{kang2008efficient}
Hyun~Min Kang, Noah~A Zaitlen, Claire~M Wade, Andrew Kirby, David Heckerman,
  Mark~J Daly, and Eleazar Eskin.
\newblock Efficient control of population structure in model organism
  association mapping.
\newblock {\em Genetics}, 178(3):1709--1723, 2008.

\bibitem{leng2012sparse}
Chenlei Leng and Cheng~Yong Tang.
\newblock Sparse matrix graphical models.
\newblock {\em Journal of the American Statistical Association},
  107(499):1187--1200, 2012.

\bibitem{lippert2011fast}
Christoph Lippert, Jennifer Listgarten, Ying Liu, Carl~M Kadie, Robert~I
  Davidson, and David Heckerman.
\newblock Fast linear mixed models for genome-wide association studies.
\newblock {\em Nature Methods}, 8(10):833--835, 2011.

\bibitem{segura2012efficient}
Vincent Segura, Bjarni~J Vilhj{\'a}lmsson, Alexander Platt, Arthur Korte,
  {\"U}mit Seren, Quan Long, and Magnus Nordborg.
\newblock An efficient multi-locus mixed-model approach for genome-wide
  association studies in structured populations.
\newblock {\em Nature genetics}, 44(7):825--830, 2012.

\bibitem{stegle2011}
Oliver Stegle, Christoph Lippert, Joris~M Mooij, Neil~D Lawrence, and Karsten~M
  Borgwardt.
\newblock Efficient inference in matrix-variate gaussian models with iid
  observation noise.
\newblock In {\em Advances in Neural Information Processing Systems}, pages
  630--638, 2011.

\bibitem{weir2006genetic}
Bruce~S Weir, Amy~D Anderson, and Amanda~B Hepler.
\newblock Genetic relatedness analysis: modern data and new challenges.
\newblock {\em Nature Reviews Genetics}, 7(10):771--780, 2006.

\bibitem{yu2005unified}
Jianming Yu, Gael Pressoir, William~H Briggs, Irie~Vroh Bi, Masanori Yamasaki,
  John~F Doebley, Michael~D McMullen, Brandon~S Gaut, Dahlia~M Nielsen, James~B
  Holland, et~al.
\newblock A unified mixed-model method for association mapping that accounts
  for multiple levels of relatedness.
\newblock {\em Nature genetics}, 38(2):203--208, 2005.

\bibitem{yuan2007}
Ming Yuan and Yi~Lin.
\newblock Model selection and estimation in the gaussian graphical model.
\newblock {\em Biometrika}, 94(1):19--35, 2007.

\bibitem{zhang2010}
Yi~Zhang and Jeff~G Schneider.
\newblock Learning multiple tasks with a sparse matrix-normal penalty.
\newblock In {\em Advances in Neural Information Processing Systems}, pages
  2550--2558, 2010.

\bibitem{zhao2007arabidopsis}
Keyan Zhao, Mar{\'\i}a~Jos{\'e} Aranzana, Sung Kim, Clare Lister, Chikako
  Shindo, Chunlao Tang, Christopher Toomajian, Honggang Zheng, Caroline Dean,
  Paul Marjoram, et~al.
\newblock An arabidopsis example of association mapping in structured samples.
\newblock {\em PLoS Genetics}, 3(1):e4, 2007.

\end{thebibliography}

\newpage
\section*{Appendix: Linear algebra identities}

The first lemma gives an alternate representation of the block trace:
\begin{mylem}
\begin{eqnarray} 
\tr{ \left( X \otimes I_N \right) M } = \tr{ X \mbox{tr}_P \left[ M \right] } \label{lem:1.1} \\
\tr{ \left( I_P \otimes X \right) M } = \tr{ X \mbox{tr}_N \left[ M \right] } \label{lem:1.2}
\end{eqnarray}
\end{mylem}
\begin{proof}
\begin{align*}
\tr{ \left( X \otimes I \right) M }	&= \sum_{a,b = 1}^P \sum_{i,j = 1}^N \left( X^T \otimes I^T \right)_{ab:ij} M_{ab:ij} \\
	&= \sum_{a,b = 1}^P \sum_{i,j = 1}^N \left( X_{ab}^T I_{ij} \right) M_{ab:ij} \\
	&= \sum_{a,b = 1}^P  \left( X_{ab}^T \right) \tr{ M_{ab:..} } \\
	&=\tr{ X \mbox{tr}_P \left[ M \right] }
\end{align*}
Similarly,

\begin{align*}
\tr{ \left( I \otimes X \right) M }	&= \sum_{a,b = 1}^P \sum_{i,j = 1}^N \left( I^T \otimes X^T \right)_{ab:ij} M_{ab:ij} \\
	&= \sum_{a,b = 1}^P \sum_{i,j = 1}^N \left( I_{ab}^T X_{ij}^T \right) M_{ab:ij} \\
	&= \sum_{i,j = 1}^N  \left( X_{ij}^T \right) \tr{ M_{..:ij} } \\
	&=\tr{ X \mbox{tr}_N \left[ M \right] }
\end{align*}

\end{proof}

\noindent As a corollary,
\begin{mylem}
If $U$ is an orthogonal matrix,
\begin{eqnarray} 
\mbox{tr}_P \left[ \left( Q \otimes U \right) W  \left( Q \otimes U \right) ^T\right] = Q \mbox{tr}_P \left( W \right) Q^T \label{lem:2}
\end{eqnarray}
\end{mylem}
\begin{proof}

\begin{align*}
\mbox{tr}_P \left[ \left( Q \otimes U \right) W \left( Q \otimes U \right)^T \right]_{ij} &= \tr{ \left( Q_i \otimes U \right) W \left( Q_j \otimes U \right)^T } \\
	&= \tr{ \left( (Q_j^T Q_i) \otimes I \right) W } \\
	&= \tr{  (Q_j^T Q_i) \mbox{tr}_P \left( W \right) } \\
	&=Q_i \mbox{tr}_P \left( W \right) Q_j^T
\end{align*}

\end{proof}

The partial trace is no longer cyclic. However, it is true that
\[ \mbox{tr}_P \left( (A \otimes I) B \right) = \mbox{tr}_P \left( B (A \otimes I) \right) \]

\noindent Another useful computation enables the block trace of $\Sigma$, which is the inverse of a sum of Kronecker products, to be taken without ever evaluating $NP \times NP$ matrices.
\begin{mylem}
\begin{eqnarray}
\mbox{tr}_P \left[ \left( A \otimes I + B \otimes X \right)^{-1} \right] = A^{-1/2} Q  \Lambda Q^T A^{-1/2} \label{eq:lem-3}
\end{eqnarray}
where
\begin{align*}
Q \Lambda_1 Q^T			:&= \mbox{Spec Decomp} \left( A^{-1/2} B A^{-1/2} \right) \\
\Lambda_2	:&= \mbox{Eigenvalues} \left( X \right) \\
\Lambda	:&= \mbox{tr}_P \left( \left[ I + \Lambda_1 \otimes \Lambda_2 \right]^{-1} \right) \\
\end{align*}
\end{mylem}
\begin{proof}

Let $Q \Lambda_1 Q^T$ be the eigendecomposition of $A^{-1/2} B A^{-1/2} $, and let $U \Lambda_2 U^T$ be the eigendecomposition of $X$. Also define $\Lambda' = \Lambda_1 \otimes \Lambda_2$. Then
\begin{align*}
\left( A \otimes I + B \otimes X \right)^{-1}	&= \left(  A^{-1/2} \otimes I  \right)  \left[ I + \left( A^{-1/2} B A^{-1/2} \right) \otimes X \right]^{-1} \left(  A^{-1/2} \otimes I  \right) \\
	&= \left(  A^{-1/2} Q \otimes U \right) \left[ I + \Lambda' \right]^{-1} \left( Q^T A^{-1/2} \otimes U^T \right) \\
\end{align*}
Defining $T = A^{-1/2} Q$ and $\Lambda = \mbox{tr}_P \left(\left[ I + \Lambda' \right]^{-1} \right) = \mbox{tr}_P \left(\left[ I + \Lambda_1 \otimes \Lambda_2 \right]^{-1} \right)$
\begin{align*}
\mbox{tr}_P \left( \left( A \otimes I + B \otimes X \right)^{-1} \right)_{ij} &= \tr{ \left[ \left(  T \otimes U \right) \left[ I + \Lambda' \right]^{-1} \left(T^T \otimes U^T \right) \right]_{[i,j]} } \\
	&= \tr{ \left(  T_{i,} \otimes U \right) \left[ I + \Lambda' \right]^{-1} \left( (T_{j,})^T \otimes U^T \right) } \\
	&= \tr{ \left( \left( T_{j,}^T T_{i,} \right) \otimes I \right) \left[ I + \Lambda' \right]^{-1} } \\
	&= \tr{ \left( T_{j,}^T T_{i,} \right) \mbox{tr}_P \left(\left[ I + \Lambda' \right]^{-1} \right) }  \tag{by Lemma 1} \\
	&= T_{i,} \Lambda T_{j,}^T \implies \\
\mbox{tr}_P \left( \left( A \otimes I + B \otimes X \right)^{-1} \right) &= T \Lambda T^T \\
	&= A^{-1/2} Q \Lambda Q^T A^{-1/2}
\end{align*}

\end{proof}

\noindent Block traces of outer products can also be efficiently computed.
\begin{mylem}
\begin{eqnarray} 
\mbox{tr}_P \left( u v^T \left[ I \otimes X \right]  \right) = \left( V^T X U \right)^T = U^T X^T V \label{lem:4}
\end{eqnarray}
\end{mylem}
\begin{proof}

Define $U = \mbox{vec}^{-1}(u)$, where $\mbox{vec}^{-1}$ maps the $NP$ vector $u$ to an $N \times P$ matrix $U$ by filling it column-wise. Define the $i$th multi-index 
\[ [i] := {(i-1)*N+0:(N-1)} \]
Then $u_{[i]} = U_{,i}$, and
\begin{align*}
\mbox{tr}_P \left( u v^T \left[ I \otimes X \right] \right)_{ij} &= \tr{ u_{[i]} v_{[j]}^T X } = \tr{ U_{,i} V_{,j}^T X } = V_{,j}^T X U_{,i} \implies \\
\mbox{tr}_P \left( u v^T \left[ I \otimes X \right]  \right) &= \left( V^T X U \right)^T = U^T X^T V
\end{align*}
\end{proof}

\noindent The next lemma computes block traces of $\mu$ quadratic forms.
\begin{mylem}
Define $T = \left[ I + (A^{-1}B) \otimes X \right]^{-1}$ for invertible matrices $A$, $B$ and $X$. Then
\begin{eqnarray}
\mbox{tr}_p \left( T y y^T T^T \left( I \otimes X \right) \right)	&=&  \left( A^{-1/2} Q S^T \right) \Lambda_X \left( A^{-1/2} Q S^T \right)^T  \label{eq:mu-id1} \\
\mbox{tr}_p \left( T y y^T T^T \right)	&=&  \left( A^{-1/2} Q S^T \right) \left( A^{-1/2} Q S^T \right)^T  \label{eq:mu-id2}
\end{eqnarray}
This uses the eigendecompositions
\[ Q \Lambda Q^T = A^{-1/2} B A^{-1/2} \mbox{ and }  U \Lambda_X U^T = X \]
and
\begin{eqnarray}
S	&:=& \mbox{vec}^{-1} \left( \diag{ \left[ I + \Lambda \otimes \Lambda_X \right]^{-1} } \right) \ast \left( U^T Y A^{1/2} Q \right) \label{eq:S}
\end{eqnarray}

\end{mylem}
\begin{proof}

First, simplify $Ty$:
\begin{align*}
Ty	&= \left[ I + (A^{-1}B) \otimes X \right]^{-1} y	\\
	&= \left(A^{-1/2} \otimes I\right) \left[ I + \left(A^{-1/2} B A^{-1/2} \right) \otimes X \right]^{-1} \left(A^{1/2} \otimes I\right) y \\
	&= \left(A^{-1/2} \otimes I\right) \left(Q \otimes U \right) \left[ I + \Lambda \otimes \Lambda_X \right]^{-1} \left(Q \otimes U \right)^T \left(A^{1/2} \otimes I\right) y \\
	&= \left(  \left( A^{-1/2} Q \right) \otimes U \right) s  \tag{+} \\
s	:&= \diag{ \left[ I + \Lambda \otimes \Lambda_X \right]^{-1} } \ast \mbox{vec} \left( U^T Y A^{1/2} Q \right)
\end{align*}

Let $X' = U D U^T$ for arbitrary $D$. Then
\begin{align*}
\mbox{tr}_p &\left( Ty y^T T^T \left( I \otimes X' \right) \right) \\
	&= \mbox{tr}_p \left[ \left(  \left( A^{-1/2} Q \right) \otimes U \right) s s^T \left( \left( Q^T A^{-1/2} \right) \otimes U^T \right) \left( I \otimes (UDU^T) \right) \right]	\tag{by (+)} \\
	&= \mbox{tr}_p \left[ \left(  \left( A^{-1/2} Q \right) \otimes U \right) s s^T \left( I \otimes D \right) \left( \left( Q^T A^{-1/2} \right) \otimes U^T \right)  \right] \\
	&=  A^{-1/2} Q \left[ \mbox{tr}_P \left( s s^T \left( I \otimes D \right) \right) \right] Q^T A^{-1/2} \tag{by \eqref{lem:2}} \\
	&=  A^{-1/2} Q \left[ S^T D S \right] Q^T A^{-1/2} \tag{by \eqref{lem:4}}
\end{align*}
Now line \eqref{eq:mu-id1} follows from taking $D= \Lambda_X \iff X' = X$ and \eqref{eq:mu-id2} follows from taking $D=I \iff X'=I$.

\end{proof}

\end{document}